\documentclass[a4paper, 12pt, oneside]{article}
\usepackage{amsmath}
\usepackage{amssymb}
\usepackage{enumerate}
\usepackage{graphicx}
\usepackage{epsf}
\usepackage{verbatim}
\usepackage{tikz}
\usepackage{systeme,mathtools}
\usepackage{float} 
\usepackage{mathrsfs}
\usepackage{csvsimple}
\usepackage{hyperref}

\usepackage{algorithm}
\usepackage{algorithmic}

\newcommand{\ENDPROCEDURE}{\STATE \textbf{end procedure}}

\usetikzlibrary{matrix,arrows,decorations.pathmorphing}

\usepackage[amsmath,thmmarks]{ntheorem}
\newcounter{thm} 

\newtheorem{Thm}[thm]{Theorem}

\newtheorem{Prop}[thm]{Proposition}
\newtheorem{Cor}[thm]{Corollary}

\theoremstyle{nonumberplain}
\theoremheaderfont{\normalfont\itshape}
\theorembodyfont{\normalfont}
\theoremseparator{.}
\theoremsymbol{\ensuremath{\square}}
\newtheorem{proof}{Proof}

\def \Z {\mathbb Z}

\def \I {\mathcal{I}}

\begin{document}

\title{Optimization of the directed spanning trees using the weighted matroid intersection algorithm} 
\author{Binhong Jiang, Gehao Wang}
\date{}
\maketitle

\begin{abstract}
In this paper, we consider the problem of updating the directed minimum spanning tree (DMST), when the given sample tree is subject to the weight changes, edge deletions and edge insertions. We present an implementation for updating the tree to a DMST using the weighted matroid intersection algorithm. Our algorithm focuses on maintaining a dynamic auxiliary graph, which plays a central role in the matroid intersection algorithm, and governs the iterations from the given tree to a DMST. Each iteration is guaranteed to yield an improved solution. We also provide an implementation of this algorithm and some experimental analysis.
\end{abstract}


\section{Introduction}
Let $G=(V,A)$ be a directed graph (digraph), formed by the vertex set $V$ and directed edge set $A$. We will call the directed cycles in $G$ dicycles. A directed spanning tree  $T$ in $G$ rooted at a vertex $r$, also known as the $r$-arborescence, is a connected acyclic subgraph of $G$ with vertex set $V$, such that $r$ has no incoming edge and every other vertex has exactly one. Once we define a weight (cost) function on the directed edges of $G$, the graph $G$ becomes a weighted digraph. Then, we say that $T$ is a directed minimum spanning tree (DMST) in $G$, if the total weight of edges of $T$ is minimum. We will assume that the weights are integers. 

The computation of the directed minimum spanning tree has been extensively studied ever since the independent work in \cite{CL}, \cite{E} and \cite{Bo}, usually referred to as Edmonds' algorithm, (see also \cite{K}).  It is well-known that finding a DMST can be solved in polynomial time. The improvements of Edmonds' algorithm on the complexity include Tarjan's algorithm in \cite{T} and an implementation of Fibonacci heaps in \cite{GGST}. Following the literature, there are many implementations and variants of the previously mentioned algorithms, such as \cite{L}, \cite{BKW}, \cite{FO}, \cite{EFRRV} and other related articles.

The matroid intersection algorithm is another approach to solve the DMST problem, see, e.g., \cite{La}, \cite{E1}, and \cite{Fr}. This algorithm uses an auxiliary graph to keep track of the edge inserting and exchanging, and it has many advantages for solving some special problems, such as the root-location problem, (for a brief review, see, e.g., \cite{KN}). After we define the weights on the auxiliary graph appropriately, the weighted version of this algorithm can optimize the solution by constantly finding paths in the auxiliary graph. The running time is also polynomially bounded. Later, to avoid any confusion, we say that the auxiliary graph is formed by nodes and directed arcs, instead of vertices and directed edges. 

In this paper, our main interest originates from algorithms for updating a DMST. A (fully) dynamic graph algorithm is a data structure that supports edge deletions, edge insertions, and it answers certain queries that are specific to the problem under consideration (see, e.g., \cite{EGI},\cite{HHS} and \cite{XW}). Based on Edmonds' algorithm, the authors in \cite{PTZ} introduced a dynamic graph algorithm for DMST, when there are edges deleted from or inserted to the weighted digraph. Later in \cite{EFRRV}, such problem was revisited, where some more experimental results and implementation details were presented. There are many practical applications of dynamically updating DMST in network optimization and hardware design, (see, e.g.,\cite{BA},\cite{KJ},\cite{PMW}). To the best of our knowledge, such updating algorithms using the matroid approach have not been paid much attention. 

Inspired by the work in \cite{BCG} and \cite{CG}, we discuss the problem of optimizing a directed spanning tree using the matroid intersection algorithm. We notice that, by detecting the existence of dicycles with negative weights in the auxiliary graph, one can find out which edges should be replaced in order to obtain an improved solution with less weight. Therefore, we give an implementation of this optimization algorithm that keeps doing the node exchange iteratively. In each iteration, the nodes being exchanged are in a negative dicycle, and they correspond to the directed edges that should be replaced in the tree $T$. In other words, this algorithm maintains the dynamic auxiliary graph by a combination of detecting negative dicycles, exchanging nodes and re-constructing the new auxiliary graph. It can be viewed as an algorithm to generate a sequence of directed spanning tree solutions in the order of decreasing cost, in constrast to the case of undirected minimum spanning tree introduced in \cite{SJ}. In addition, the dynamic algorithm introduced in \cite{PTZ} needs to be implemented with Edmonds' algorithm, so that the edge contraction information of the optimal tree can be recorded in an augmented data structure during the optimizing process. In our case, we aim to propose an alternative and independent approach to optimize sample trees without the original generating algorithm. Our implementation is available at \url{https://github.com/heitopai/exp} (accessed on 30 January 2026).

The structure of this paper is as follows. In Sect.\ref{S2}, we give some background on the weighted matroid intersection algorithm. Sect.\ref{S3} includes the detailed explanation of our algorithm for improving the directed spanning tree. In Sect.\ref{S4}, we present the complete description of the directed spanning tree optimization, and discuss the complexity. Finally, we give our experimental investigation in Sect.\ref{S5}.

\section{Weighted matroid intersection algorithm}\label{S2}
In this section, we briefly review the weighted matroid intersection algorithm. For more details, we refer to \cite{BA}, \cite{CCPS}, \cite{O}, \cite{S} and other related materials. 

Generally, for two matroids $(S,\I_{1})$ and $(S,\I_{2})$ defined on a finite set $S$, the collection $\I_1\cap \I_2$ consists of common independent sets. The $k$-intersection is an independent set of size (cardinality) $k$. The maximum-size common independent set problem can be solved using an augmenting path algorithm in order to find a set with maximum size, where the {\bf auxiliary graph} plays a central role.

Let $I$ be an independent set in $\I_1\cap \I_2$. The auxiliary graph $D(I)$ (also called exchangeability graph in \cite{BKYY}) corresponding to $I$ is a bipartite digraph with node set $S$ defined as follows. For any $y\in I$ and $x\in S\setminus I$,
\begin{align*}
     & \mbox{there is an arc } (y,x) \mbox{ if and only if } I-y+x\in \I_1;\\
     & \mbox{there is an arc } (x,y) \mbox{ if and only if } I-y+x\in \I_2.
\end{align*}
The directed arc describes the consequence of exchanging $x$ and $y$. 

In our context, we work on the directed spanning trees. Consider the digraph $G=(V,A)$ with no loop. We denote $\delta(v)$ to be the set of incoming edges of vertex $v$ in $G$. The first matroid we consider is the {\bf graphic matroid} $(A,\I_1)$, where $\I_1$ consists of all subsets $X$ of $A$ such that $(V,X)$ has no cycles, disregarding the directions. The second one is the {\bf partition matroid} $(A,\I_2)$, where $\I_2$ consists of all subsets $X$ of $A$ such that, for every $v\in V$ and $v\neq r$,
\begin{equation*}
	|\delta(v)\cap X|\leq 1 \quad\mbox{and}\quad |\delta(r)\cap X|=0.
\end{equation*}
Here the root vertex $r$ does not have to be the same for all independent sets of $\I_2$. For example, in Figure \ref{E1}, we can see that $\left\{y_1,y_2,y_3\right\}$ is a common independent set, which is a directed spanning tree of the digraph in Figure \ref{graph:G}.
\begin{figure}[H]	
\begin{minipage}{0.5\textwidth}
	\centering
	\includegraphics[width=5.5cm]{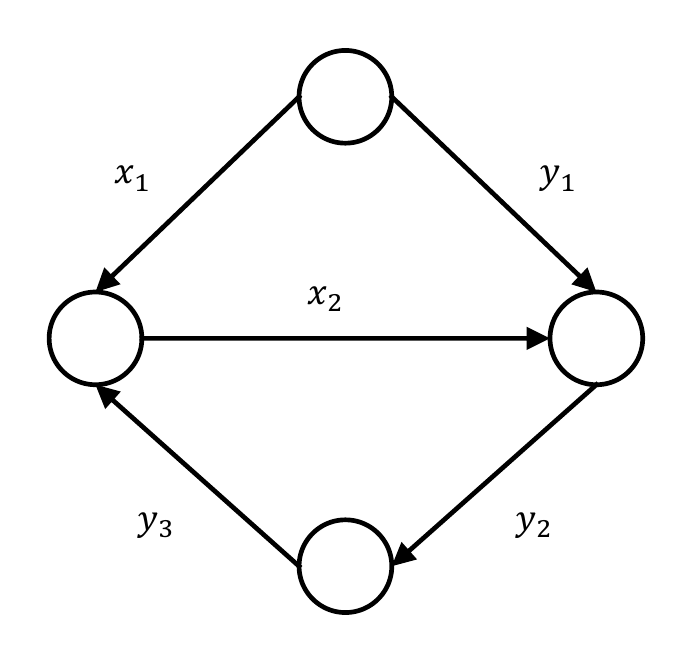}
		\caption{A digraph $G$ }\label{graph:G}
\end{minipage}
\begin{minipage}{0.5\textwidth}
	\centering
\includegraphics[width=5.5cm]{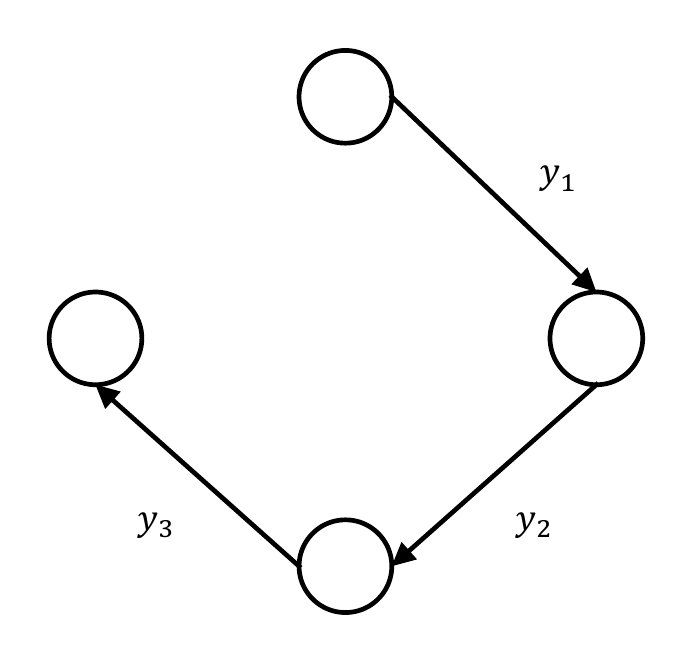}
\caption{A directed spanning tree for $G$}\label{E1}  
\end{minipage}
 \end{figure}
The auxiliary graph for the digraph $G$ is presented in Figure \ref{E2}. The arc $(y_1,x_1)$ means that, after replacing $y_1$ with $x_1$, the set $\left\{x_1,y_2,y_3\right\}$ still forms a spanning tree. The arc $(x_2,y_1)$ indicates that the set $\left\{x_2,y_2,y_3\right\}$ forms a subgraph, where each vertex except the root has exactly one incoming edge.
\begin{figure}[H]	
	\centering
	\includegraphics[width=5cm]{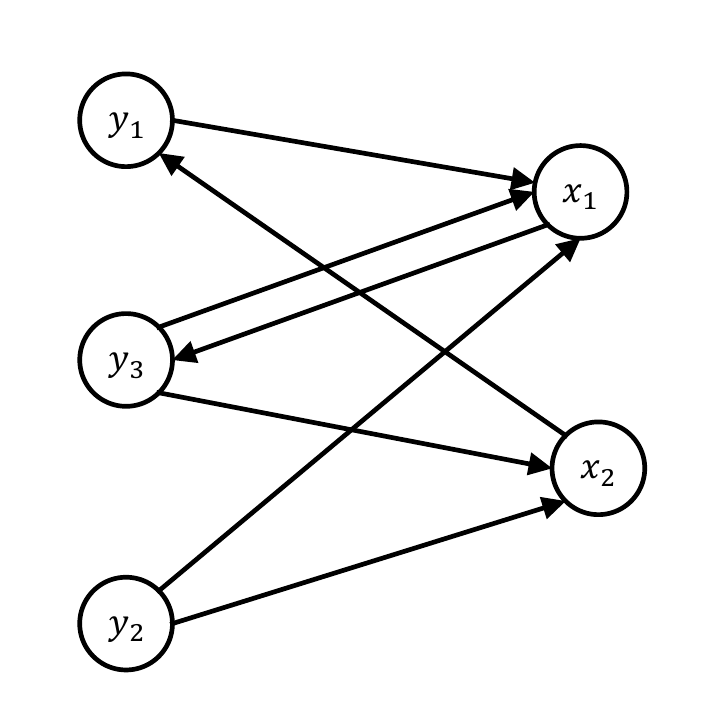}
	\caption{The auxiliary graph}\label{E2}  
\end{figure}
 
Usually, in order to obtain an independent set $I$ of maximum size, we can perform the following algorithm with the input $I$, which can be initialized as $I=\emptyset$. Consider the two subsets in $S\setminus I$,
\begin{align*}
	&X_1:=\{x\in  S\setminus I \,|\,I+x\in \I_1\}, \\
	&X_2:=\{x\in  S\setminus I \,|\,I+x\in \I_2\}.
\end{align*}
Suppose there is a path in the auxiliary graph $D(I)$ from a node $x_1$ in $X_1$ to a node in $X_2$. We take a shortest such path 
\begin{equation*}
P=\{x_1,y_1,x_2,y_2,\dots,x_n,y_n,x_{n+1}\} \quad\mbox{with}\quad x_{n+1}\in X_2, \,n\geq 0.
\end{equation*}
Let $I^{'}=I\bigtriangleup P$ be the output, where the notation $\bigtriangleup$ means the {\bf symmetric difference}:
\begin{equation*}
	I\bigtriangleup P=\left(I\setminus \{y_1,\dots,y_n\}\right)\cup \{x_1,\dots,x_{n+1}\}. 
\end{equation*} 
Then $|I^{'}|=|I|+1$. If there is no such path, then $I$ is already a maximum-size common independent set.

The weighted matroid intersection algorithm is performed on the auxiliary graph with a weighted setting. For example, let $w:S\rightarrow \Z$ be a weight function, and $I$ be an independent set.  We define
\begin{equation*}
	w(I)=\sum_{y\in I}w(y).
\end{equation*}
The maximal-weight common independent set problem is therefore to look for a set $I$ with $w(I)$ being the maximal. We can define a cost function on the nodes of the auxiliary graph $D(I)$ as
\[
c(e)=
\begin{cases}
	w(e) & e\in I ; \\
	-w(e)  & e\in S\setminus I.
\end{cases}
\]
The cost of a path $P$ is $c(P)=\sum_{e\in P}c(e)$. In this way, the maximal $k$-intersection problem can be solved via augmentation along shortest-cheapest path. That is, among all the minimum cost path $P$, we also require $P$ to be the shortest, before the augmentation $I^{'}=I\bigtriangleup P$. 

On the other hand, if $I$ is already a $k$-intersection, then the following result (see, e.g., \cite{BCG} and \cite{F}) not only allows us to check whether $w(I)$ is maximal using the auxiliary graph, but also shows when to perform the node exchange in order to optimize $I$. 

\begin{Thm}\label{BCG1}
	\begin{enumerate}[(i)]
		\item $I$ is a maximal $k$-intersection if and only if there are no negative dicycles in the auxiliary graph $D(I)$ with the cost function $c(e)$.
		\item Let $C$ be a dicycle in $D(I)$, and let $X$ and $Y$ be the node sets of $C$ in  $S\setminus I$ and $I$ respectively. If $I'=I-Y+X$ is not a $k$-intersection, then $C$ contains a sub-dicycle (a dicycle on a proper subset of the nodes of $C$).
		\item Let $C$ be a negative dicycle in $D(I)$ that contains no negative sub-dicycle. Then $I'=I-Y+X$ is a $k$-intersection with $w(I')>w(I)$.
	\end{enumerate}
\end{Thm}
Note that the statement $(ii)$ above implies that if $I'$ is not a $k$-intersection and $w(C)<0$, then $C$ contains a negative sub-dicycle \cite{BCG}. This gives us the statement $(iii)$, which is the foundation of our algorithm.

\section{Improving the directed spanning trees}\label{S3}

Consider the digraph $G$ with a weight function $w:A\rightarrow \Z$. Finding a directed minimum spanning tree $T$ of $G$ is a minimum-cost maximum-size common independent set problem \cite{O}. By abuse of notation, we also let $T$ be the intersection of $\I_1$ and $\I_2$. The cost function of the auxiliary graph $D(T)$ we consider is 
\[
l(e)=
\begin{cases}
	-w(e) & e\in T  \\
	w(e)  & e\in A\setminus T.
\end{cases}
\]
Furthermore, following Theorem \ref{BCG1}, we have
\begin{Cor}
	A subgraph $T$ of $G$ is a directed minimum spanning tree rooted at $r$ if and only if $T\in \I_1\cap\I_2$, $|T|=|V|-1$ and there are no negative dicycles in $D(T)$ with the cost function $l(e)$. 
\end{Cor}
Therefore, to check whether $T$ is still a DMST after some weight changes, edge deletion or edge insertion, we can simply detect whether there exists a negative dicycle in the auxiliary graph. Later, we will discuss how to obtain an improved solution $T'$ with less weight. 

\subsection{Construct the auxiliary graph $D(T)$}

We first discuss how the auxiliary graph $D(T)$ can be constructed. (In this case, for $A\setminus T$, the subsets $X_1$ and $X_2$ defined in Sect. \ref{S2} are empty). Let us consider the following two kinds of subsets:
\begin{align}
	&N(T,x)=\{y\in T\,\,|\,\,T-y+x\in \I_{1}\};\nonumber\\
	&N(T,y)=\{x\in A\setminus T\,\,|\,\,T- y+x\in \I_{2}\} \label{R}.
\end{align}
In other words, the subset $N(T,x)$ is the set of neighbors of $x$ that belong to the directed arcs pointing at $x$, and $N(T,y)$ belong to the arcs pointing at $y$. By the definition of the auxiliary graph in Sect.\ref{S2}, all the directed arcs of $D(T)$ can be described by the following two types:
\begin{align*}
	& \I_1\mbox{ - arcs }: \left(y,x\right), \mbox{ for all } x\in A\setminus T \mbox{ and }y\in N(T,x),\\
	& \I_2\mbox{ - arcs }: \left(x, y\right), \mbox{ for all } y\in T  \mbox{ and } x\in N(T,y).
\end{align*}
Therefore, the construction of $D(T)$ is to classify the subsets $N(T,x)$ and $N(T,y)$ for all $x$ in $A\setminus T$ and $y\in T$, respectively. 

If $e$ is a directed edge in $G$, let $h_e$ be the head of $e$ and $t_e$ be the tail of $e$, that is, $e=(t_e,h_e)$. Then, 
\begin{Prop}\label{I1I2}
	\begin{enumerate}[(i)]
		\item For any $x$ in $A\setminus T$, the subset $N(T,x)$ correspond to the edges on the unique path between the vertices $h_{x}$ and $t_{x}$ in $T$ of $G$, disregarding the directions.
		\item For any $y\in T$, $N(T,y)=\delta(h_y)-\{y\}$. 
	\end{enumerate}
\end{Prop}
\begin{proof}
For $x\in A\setminus T$, there is a unique path between $h_{x}$ and $t_{x}$ in $T$. Adding the edge $x$ to $T$ will give us a cycle, if we disregard the directions. Therefore,  $y$ must be an edge on the unique path between $h_{x}$ and $t_{x}$, in order to make sure that $T-y+x$ is in $\I_1$. This proves the first statement. The definition of $N(T,y)$ in Equation \eqref{R} immediately implies that the subset $N(T,y)$ corresponds to the edges $x$ in $A\setminus T$ with $h_x=h_y$. This proves the second statement.
\end{proof}

A common approach to find the path between two vertices in a tree is to use the Lowest Common Ancestor (LCA) algorithm (see, e.g. \cite{HT}). Let $P_e$ be the unique path in $T$, formed by a sequence of vertices, corresponding to a non-tree edge $e$. We first traverse upwards from $t_e$ to the root $r$, marking all visited vertices. Then we traverse upwards from $h_e$ until encountering the first marked vertex, which is the lowest common ancestor. So the procedure is $\text{LCA}(T, r, t_e, h_e)$. In this way, we can collect all the edges on the path $P_e$. The construction of $\I_1$ - arcs can be done using the following algorithm.
\begin{algorithm}[H]
	\caption{Construction of $\I_1$ - arcs}
	\begin{algorithmic}[1]\label{alg:I1}
		\REQUIRE Directed graph $G=(V,A)$, directed spanning tree $T$ rooted at $r$
		\ENSURE $\I_1$ - arcs
		\FOR{each edge $x \in A \setminus T$}
		\STATE $\text{lca} \gets \text{LCA}(T, r, t_x, h_x)$
		\STATE $\text{cur} \gets t_x$
		\WHILE{$\text{cur} \neq \text{lca}$}
		\STATE $y \gets \text{incoming edge of cur in } T$
		\STATE Create directed arc $(y, x)$
		\STATE $\text{cur} \gets \text{parent of cur in } T$
		\ENDWHILE
		\STATE $\text{cur} \gets h_x$
		\WHILE{$\text{cur} \neq \text{lca}$}
		\STATE $y \gets \text{incoming edge of cur in } T$
		\STATE Create directed arc $(y, x)$
		\STATE $\text{cur} \gets \text{parent of cur in } T$
		\ENDWHILE
		\ENDFOR
	\end{algorithmic}
\end{algorithm}

From Proposition \ref{I1I2}, we can see that the construction of $\I_2$ - arcs is very straightforward.
\begin{algorithm}[H] 
	\caption{Construction of $\I_2$ - arcs}
	\begin{algorithmic}[1]\label{alg:I2}
		\REQUIRE Directed graph $G=(V,A)$, directed spanning tree $T$ rooted at $r$
		\ENSURE $\I_2$ - arcs
		\FOR{each edge $y \in T$}
		\FOR{each edge $x \in \delta(h_y)-\{y\}$}
		\STATE Create directed arc $(x, y)$
		\ENDFOR
		\ENDFOR
	\end{algorithmic}
\end{algorithm}

\subsection{Simple negative dicycle}

We call a negative dicycle $C_0$ in $D(T)$ the {\bf simple negative dicycle} if it contains no negative sub-dicycle. By Theorem \ref{BCG1}, we can see that $T'=T\bigtriangleup C_0$ is guaranteed to be a directed spanning tree. Let $C_0$ be a simple negative dicycle in $D(T)$ with
\begin{equation}\label{C0}
	C_0=\{y_{1},x_{1},y_{2},x_{2},\dots y_{n}, x_{n}, y_{n+1}=y_{1}\}, 
\end{equation}
and $y_{i}\in T, x_{i}\in A\setminus T,$ for $i=1,\dots, n$. Removing the edges $\{y_{1},\dots y_{n}\}$ in $T$ will divide the tree $T$ into $n+1$ sub-trees. We denote $H_i$ to be the sub-tree containing the head of the directed edges $y_i$ for $1\leq i\leq n$, and the last one to be $H_{n+1}$. The arc $(x_i,y_{i+1})$ in the dicycle $C_0$ implies that  $y_{i+1}$ and $x_i$ share the same head. So the directed spanning tree $T'$ is formed by the sub-trees
\begin{equation*}
	x_n\cup H_1,\quad x_1\cup H_2,\quad \dots,\quad x_{n-1}\cup H_n \quad\mbox{and}\quad H_{n+1}.
\end{equation*}
Since $w(C_0)<0$, by the definition of the cost function $l(e)$, 
\begin{equation*}
	\sum_{i=1}^nw(x_i)<\sum_{i=1}^nw(y_i).
\end{equation*}
This means that we just obtained a directed spanning tree $T'$ with less weight. Therefore, in order to find an improved solution $T'$ with $w(T')<w(T)$, we first need to detect the simple negative dicycle in the auxiliary graph $D(T)$. 

The negative cycle problem usually combines a single source shortest path algorithm and a cycle detection strategy \cite{CG}. The single source shortest path algorithm is to find a path with minimum weight (cost) between two vertices in a graph. All the previously known classical algorithms are based on the labeling method. Dijkstra's algorithm \cite{D} is unable to handle negative weights. Other algorithms that can be used for detecting negative cycles are usually built on the famous Bellman-Ford-Moore algorithm (see, e.g., \cite{B}, \cite{FF}, \cite{M}).

For the source node (or vertex) $s$, the labeling method maintains a distance label $d(v)$ and parent pointer $p(v)$ for every node $v$ with the initiation $d(v)=\infty$, $p(v) =  \mbox{null}$, $d(s) = 0$. At each step, we select an arc $(u, v)$ satisfying $d(u)<\infty$ and $d(u)+w(u,v)<d(v)$. Then we set $d(v)=d(u)+w(u,v)$ and $p(v)=u$. We keep repeating this operation until no such arcs can be selected. The cycle detection strategies for the labeling method use the so-called parent graph. Namely, for a directed graph $D$, the parent graph $D_p$ is the dynamic graph induced by the arcs $(p(v), v)$ with $p(v)\neq \mbox{null}$ during the operation. If $D$ contains a negative cycle, then after a finite number of labeling operations, $D_p$ always has a cycle. In our case, we want to perform the \textbf{immediate cycle detection} on the directed graph. Namely, we want to detect a negative dicycle in $D_p$ when it first appears. Therefore, after applying the labeling operation on an arc $(u,v)$, if $p(v)=u$ and $v$ already appears in the previous parent pointers, then we have found a negative dicycle. Tarjan’s subtree disassembly strategy \cite{T2} is very suitable for our case. For more details about the negative cycle problem, we refer to \cite{CG}, \cite{MARRB} and references therein.

However, the usual labeling operation can not guarantee that the negative dicycle we find is a simple negative dicycle. It maintains the set of labeled nodes in a first-in, first-out (FIFO) queue, and there is usually no way to predict the next labeled node that will appear in our simple negative dicycle. Suppose $C$ is a negative dicycle in a weighted directed graph $D$. We aim to check whether $C$ is a simple negative dicycle. If not, we want to find the sub-dicycle in $C$ being a simple negative dicyle. Let us consider the subgraph $D_C$ formed by the nodes of $C$ and the arcs in $D$ whose heads and tails are also in $C$. We denote $D_C\setminus \{v\}$ to be the subgraph of $D_C$ after removing the node $v\in D_C$ and its adjacent arcs. If $D_C\setminus \{v\}$ contains no negative dicycle for all $v\in C$,  then $C$ is already a simple negative dicycle. On the other hand, if we find a negative dicycle $C'$ in $D_C\setminus \{v\}$, then we repeat this argument on $D_{C'}$, which is a subgraph of $D_C$. This procedure can be described as follows.

\begin{algorithm}[H]		
	\caption{{\sf : Procedure }NegativeSubDicycle($C$)}\label{alg:sub}
	\begin{algorithmic}[1]
		\FOR{each node $v\in C$}
		\STATE $D_C\setminus \{v\} \leftarrow$ the subgraph of $D_C$ after removing $v$ and arcs adjacent to $v$
		\IF{find a negative dicycle $C'$ in $D_C\setminus \{v\}$}
		\RETURN $C'$
		\ENDIF
		\ENDFOR
		\RETURN ``no negative sub-dicycle''
		\ENDPROCEDURE
	\end{algorithmic}
\end{algorithm}

\noindent
Eventually, we will find the simple negative dicycle $C_0$ in $C$ by traversing all vertices of $C$ at most once. It leads us to the following algorithm.

\begin{algorithm}[H]
	\caption{Simple Negative Dicycle Algorithm}\label{alg:simple}
	\begin{algorithmic}[1]
		\REQUIRE A weighted directed graph $D$.
		\ENSURE A simple negative dicycle or no ``no negative dicycle''.
		
		\STATE Run negative cycle detection on $D$
		\IF{find a negative dicycle $C$}
		\WHILE{NegativeSubDicycle($C$) return $C'$} 
		\STATE $C \leftarrow C'$ 
		\ENDWHILE
		\RETURN $C$
		\ELSE
		\RETURN ``no negative dicycle''
		\ENDIF
	\end{algorithmic}
\end{algorithm}

\section{Optimization of the directed spanning tree}\label{S4}
In this section, we present our algorithm for the optimization of a directed spanning tree $T$ by combining the algorithms introduced in the previous section. 

Let $C_0$ be the simple negative dicycle presented in Equation \eqref{C0}, where $x_i\in T'$ and $y_i\in A\setminus T'$. After the node exchange process $T'=T\bigtriangleup C_0$, we need to update the auxiliary graph $D(T)$ to $D(T')$. The $\I_1$ - arcs in the auxiliary graph $D(T')$ can be done using Algorithm \ref{alg:I1}. That is, we need to compute all the edges on the unique path between the vertices $h_{e}$ and $t_{e}$ in $T'$ of $G$, for every non-tree $e$ in $G$. For the $\I_2$ - arcs in $D(T')$, the arcs $(x_i,y_{i+1})$ in the dicycle $C_0$ implies that $h_{x_i}=h_{y_{i+1}}$. If we define a bijection map $f$ from the nodes of $D(T)$ to the nodes of $D(T')$ as
\begin{align*}
	&f(x)=x, \mbox{ for } x\neq x_i \quad\mbox{and}\quad f(y)=y,  \mbox{ for } y\neq y_i;\\
	&f(x_i)=y_{i+1} \quad\mbox{and}\quad f(y_{i+1})=x_i, 
\end{align*}
for $i=1,2,\dots,n$, then the $\I_2$ - arcs in $D(T')$ are the arcs $(f(x),f(y))$, where $x$ is in $N(T,y)$, for all $y\in T$. The pseudocode is the following.
\begin{algorithm}[H]
	\caption{Update Algorithm for $T$ and $D(T)$}
\begin{algorithmic}[1]	\label{alg:update}
		\REQUIRE Directed graph $G=(V,A)$, directed spanning tree $T$ with root $r$ and $D(T)$
		\ENSURE Updated $T$ and $D(T)$, where $T$ is a DMST
		\STATE $C \gets$ negative cycles in $D(T)$ detected
		\WHILE{$C \neq \emptyset$}
		\STATE $C_0 \gets$ simple negative cycles in $D(T)$ found by Algorithm \ref{alg:simple}
		\FOR{each $v \in C_0$}
		\STATE Delete all arcs in $D(T)$ containing $v$
		\IF{$v \in A \setminus T$}
		\STATE Update $T$ by replacing the edge $e$ with $v$, where $h_e=h_v$ 
		\STATE Exchange the node $e$ with $v$ in $D(T)$
		\FOR{each $x \in \delta(h_v)-v$}
		\STATE Create directed arc $(x, v)$
		\ENDFOR
		\ENDIF
		\ENDFOR
		\STATE Update $I_1$ - arcs in $D(T)$ using Algorithm \ref{alg:I1}
		\STATE $C \gets$ negative cycles in updated $D(T)$ detected
		\ENDWHILE
		\RETURN $T$ and $D(T)$
	\end{algorithmic}
\end{algorithm}

\subsection{Edge deletion and insertion}
In the presence of the auxiliary graph $D(T)$, the dynamic process can be handled in a more direct way when a directed edge is inserted to or deleted from the weighted digraph $G=(V,A)$. It is worth mentioning that multiple insertions and deletions can be done simultaneously.

Let $T$ be a DMST of $G$, and $e_{out} \in A$ be a directed edge that is about to be deleted from $G$. If $e_{out}\in A\setminus T$, then $T$ is still a DMST. If $e_{out}\in T$, then we simply change the weight of $e_{out}$ to $\infty$ in $G$, and start updating $D(T)$ from this node $e_{out}$. This is equivalent to the weight change operation. If the negative cycle detection fails at finding a dicyle involving $e_{out}$, that means there is no directed spanning tree after deleting $e_{out}$. 

Let $e_{in}$ be a new directed edge we want to insert to $G$ with a reasonable weight and $t_{e_{in}},h_{e_{in}}\in V$, such that $T$ may not be a DMST anymore. We modify the auxiliary graph $D(T)$ by first adding the new node $e_{in}$ to the $A\setminus T$ node set of $D(T)$. Then, we compute the subset $C(T,e_{in})$ using $\text{LCA}(T,r,t_{e_{in}},h_{e_{in}})$ to fill up the $\I_1$ - arcs of the new auxiliary graph, and the $\I_2$ - arcs $(e_{in},y)$ by finding the edges $y$ in $T$ with $h_y=h_{e_{in}}$. Finally, the algorithm attempts to find a simple negative dicycle in the new $D(T)$ from the node $e_{in}$. If there is no such cycle, then $T$ is still a DMST. 

For example, Figure \ref{ins1} is a digraph $G$ with weights. Let vertex \text{0} be the root. Then, the set $\{e_1,e_3,e_4\}$ forms a DMST for $G$, and Figure \ref{ins3} is the auxiliary graph.
\begin{figure}[h]	
	\begin{minipage}{0.5\textwidth}
		\centering
		\includegraphics[width=6.3cm]{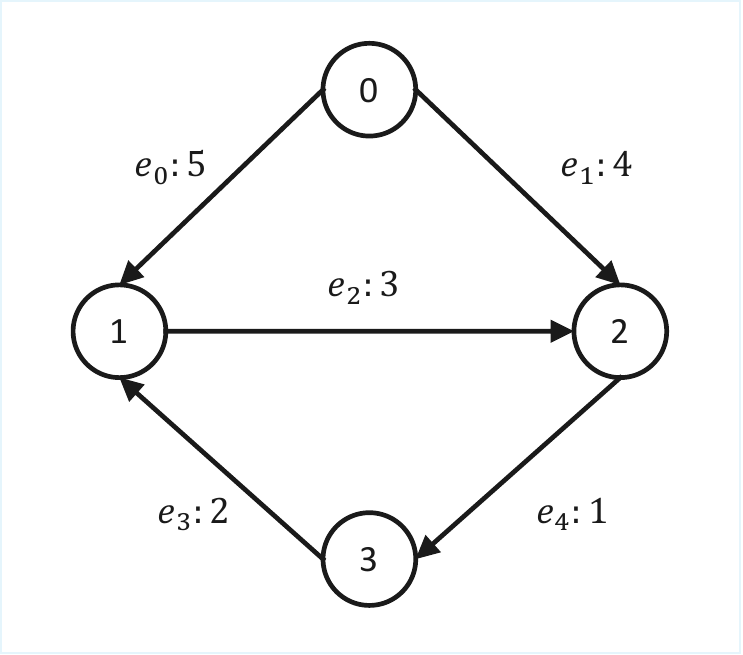}
		\caption{A digraph $G$ with weights}\label{ins1}
	\end{minipage}
	\begin{minipage}{0.5\textwidth}
			\centering
		\includegraphics[width=5.5cm]{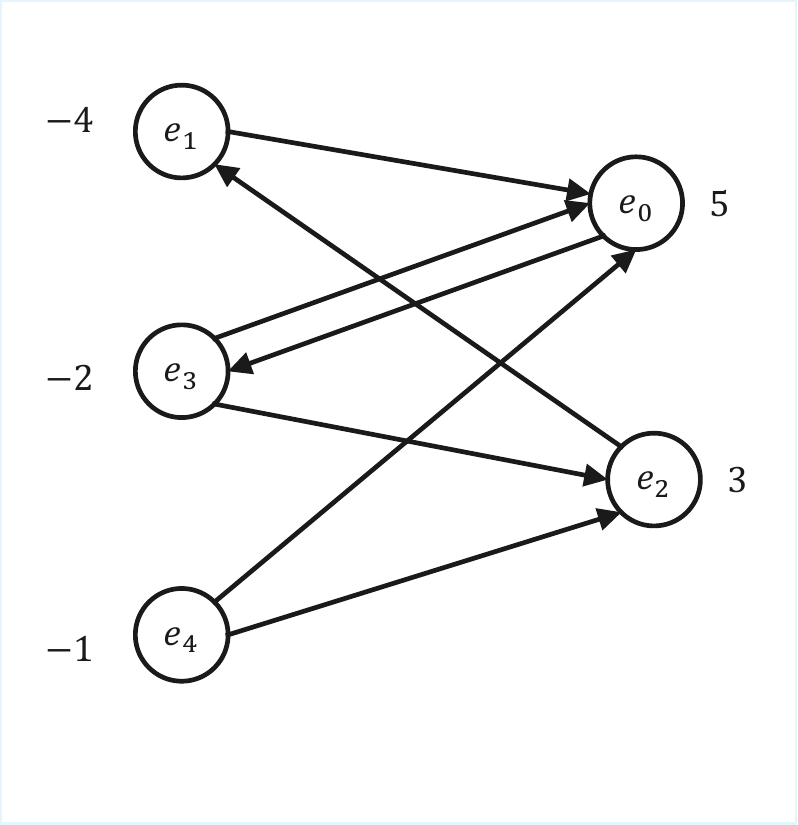}
		\caption{The original auxiliary graph }\label{ins3}  
	\end{minipage}
\end{figure}
\noindent
Let $e_6$ be a new directed edge we want to insert to $G$ with weight $1$, as shown in Figure \ref{ins4}. The head and tail of $e_6$ are vertex \text{1} and vertex \text{2}, respectively. To modify the auxiliary graph, we add the $\I_1$ - arcs $(e_3,e_6)$ and $(e_4,e_6)$ after computing the path in $T=\{e_1,e_3,e_4\}$ that connects the vertex \text{1} and vertex \text{2}, (see Figure \ref{ins5}). Since the head of $e_3$ is also vertex \text{1}, we add the $\I_2$ - arc $(e_6,e_3)$. Then, we run the detection on the cycles involving $e_6$, and find out there is a simple negative cycle $\{e_3,e_6\}$. Exchanging $e_3$ and $e_6$ will give us an improved solution. 
\begin{figure}[H]	
	\begin{minipage}{0.5\textwidth}
		\centering
	\includegraphics[width=6.3cm]{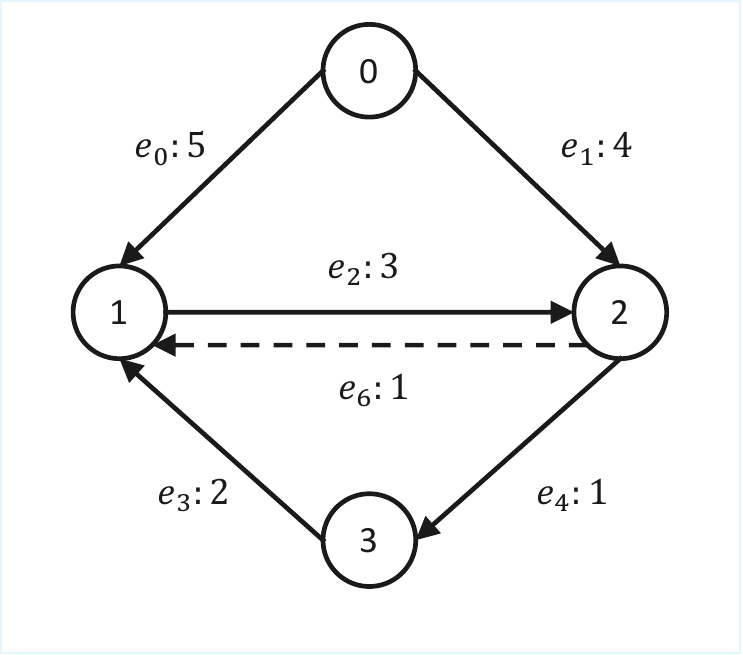}
	\caption{The insertion of $e_6$ to $G$}\label{ins4}  
	\end{minipage}
	\begin{minipage}{0.5\textwidth}
		\centering
			\includegraphics[width=5.5cm]{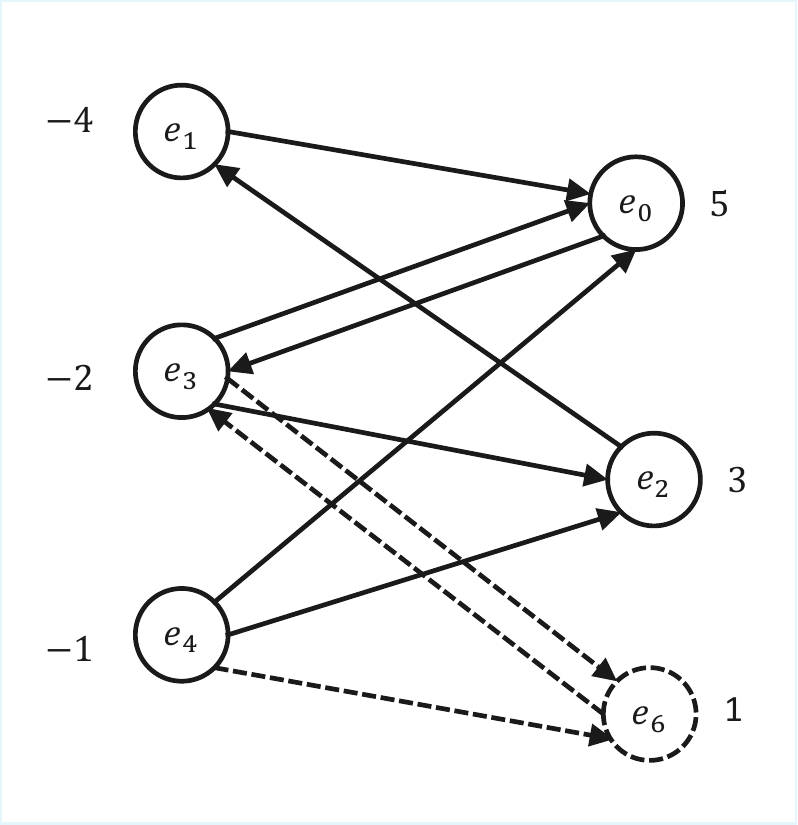}
		\caption{The new auxiliary graph}\label{ins5}  
	\end{minipage}
\end{figure}

\subsection{Complexity Analysis}

For the graph $G=(V,A)$, let $n=|V|$ and $m=|A|$. Initially, given a possible directed minimum spanning tree $T$, we construct $D(T)$. The $\I_1$-arcs in $D(T)$ can be computed using Algorithm \ref{alg:I1}. We need to compute the path in $T$ for all the $m-n+1$ non-tree edges, and the LCA procedure takes $O(n)$ time. Then, finding the $\I_1$-arcs of $D(T)$ yields an $O(mn)$ time. The number of $\I_2$-arcs in $D(T)$ is $m-n+1$. We simply check the incoming edges for each vertex of $G$ and find out which $y$ in $T$ share the same head as $x$ in $A\setminus T$. To sum up, constructing $D(T)$ can be done in $O(nm)$ time. Let $M$ be the total number of arcs in $D(T)$. The path in $T$ contains at most $n-1$ edges. Therefore, $M$ is at most $(m-n+1)n$. As mentioned in \cite{CG},  Tarjan's cycle detection runs in $O(nM)=O(n^2m)$ time.

If the algorithm finds a negative dicycle, then Algorithm \ref{alg:simple} is executed. The loops at line 3 of the Simple Negative Dicycle Algorithm in Algorithm \ref{alg:simple} needs at most $\frac{|C|}{2}$ times, ($|C|$ is an even number). Since $|C|\leq 2(n-1)$, finding a simple negative dicycle runs in $O(n^3m)$ time in the worst case. 

After the node exchange $T'=T\bigtriangleup C_0$, we can update the auxiliary graph to $D(T')$. If we keep repeating this updating operation, after a finite number of iterations, we will end up with a DMST. Let $T_o$ be the final optimality. Then, the number of iterations needed is at most $w(T)-w(T_o)$ theoretically, and the overall time complexity is $O\left((w(T)-w(T_{o}))n^{3}m\right)$. In this way, our algorithm generates a sequence of directed spanning trees with strictly decreasing cost. Practically, the procedure can be terminated early upon locating a spanning tree that satisfies specific predefined constraints. The simple negative dicycle in each iteration explicitly encodes the edge exchange operations from the current tree to the next improved solution, which can be useful in certain scenarios.

\section{Experimental Analysis}
\label{S5}

Our algorithms were implemented in Python 3.12.4, and experiments were conducted on a computer with the following hardware specifications:  Intel(R) Core(TM) i5-8400 CPU @ 2.80GHz with 16GB RAM. We performed two types of experiments. First, we generated initial directed spanning trees using a greedy algorithm and applied our matroid algorithm to test the optimization (see Algorithm \ref{alg:greedy} in Appendix A.1 for more details). The second type is our dynamic algorithm on edge deletion and insertion. The purpose of these experiments is to evaluate the correctness and the running times of our algorithm. 

For dense graphs, we set the vertex count $|V| = 50i$, where $i = 1,2,\ldots,10$. For each vertex count, we generated graphs with different edge densities by setting the edge probability $p = 0.2j$, where $j = 1,\dots,5$. For sparse graphs, we set $|V| = 100i$, where $i = 1,2,\ldots,10$. For each vertex count, we generated graphs with the edge probability $p = \frac{c}{|V|-1}$, where $c \in \{10, 20, 30, 40, 50\}$, so that the expected average degree remains constant as $|V|$ increases. The reported running time for each vertex count is the average over the generated instances.

We generate these graphs using the Erdős-Rényi model \cite{ER} via NetworkX's \texttt{fast \_gnp\_random\_graph} generator. Suppose the vertices are indexed by $\{0,1,\dots,|V| -1\}$. To ensure at least one directed spanning tree rooted at vertex $0$ exists, for each vertex $v\in \{1,2,\dots,|V|-1\}$, we add an incoming edge $(u,v)$ from a randomly selected vertex $u$ in $[0,v-1]$ to $v$. In this way, these $|V|-1$ directed edges form a connected graph, and each vertex has exactly one incoming edge. Hence it is guaranteed to be a directed spanning tree. The edge weights are uniformly sampled from integers in $[0,1000]$.

\subsection{Generating a DMST}
\label{sec:comparison}
In this subsection, we discuss the experimental results of generating directed minimum spanning trees using our matroid algorithm. We also execute NetworkX's implementation \cite{HSS} of Edmonds' algorithm \texttt{minimum\_spanning\_arborescence}, which is based on the original version \cite{E} with the complexity $O(mn)$, on the same datasets. 

By comparing the outputs, we notice that our matroid algorithm consistently computes the correct DMST for all instances. While it generally surpasses Edmonds' algorithm on dense graphs (see Figure \ref{fig:dense}), its performance gain on sparse graphs becomes increasingly pronounced with larger graph sizes, (see Figure \ref{fig:sparse}). These results indicate that our implementation of the matroid algorithm can be used as an alternative approach to compute directed minimum spanning trees. In terms of the time complexity, constructing the initial tree and the auxiliary graph runs in $O(mn)$. Therefore, the overal worst-case complexity is higher than Edmonds' algorithm. However, in comparison with the global ``contract-and-expand'' mechanism of Edmonds' algorithm, which usually causes cache misses, the negative cycle detection using Tarjan's subtree disassembly is a local operation that accesses memory in a more predictable order. We believe this explains why the matroid algorithm performs faster in our experiments. It also explains why this algorithm is very competitive particularly for sparse graphs, where $m=O(n)$. In such cases, the auxiliary graph contains fewer arcs, resulting in a limited number of node exchanges.  

\begin{figure}[h]
	\begin{minipage}{0.48\textwidth}
		\centering
		\includegraphics[width=\linewidth]{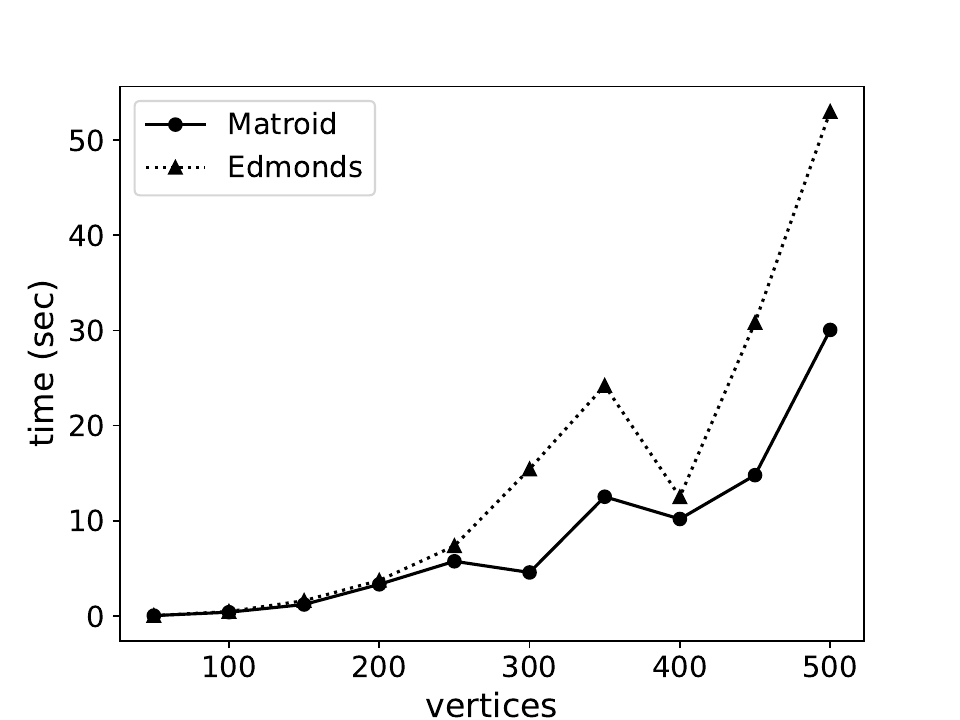}
		\caption{The comparison of running times for dense graphs}\label{fig:dense}
	\end{minipage}\hfill
	\begin{minipage}{0.48\textwidth}
		\centering
		\includegraphics[width=\linewidth]{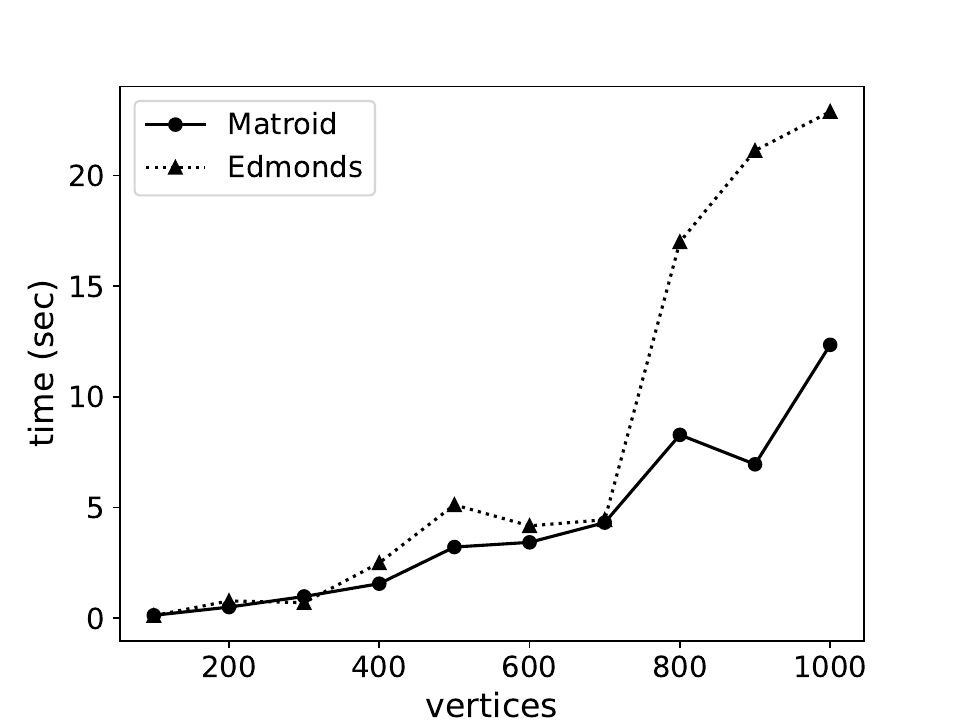}
		\caption{The comparison of running times for sparse graphs}\label{fig:sparse}
	\end{minipage}
\end{figure}

\subsection{Dynamic practice}
\label{sec:dynamic}
In the dynamic algorithm experiments, our main concern is the time savings against re-executing the static generating algorithm on the modified graph. For the edge deletion experiment, we randomly delete three edges of the original tree. For the edge insertion experiment, we randomly generate three directed edges with their heads not being the root, and the weights of inserted edges are randomly assigned values between 0 and 10. Then, we run the dynamic algorithm and our matroid algorithm on the modified graphs, and record the running times for comparison. We compute the time gain as
\begin{equation*}
	\text{time gain} = \frac{{\text{static time}} - {\text{dynamic time}}}{{\text{static time}}}\times 100\%.
\end{equation*}
Note that, we did not make comparison with the results in \cite{PTZ}, since our strategy is different and they did not provide a publicly available implementation. Their algorithm needs to be implemented with Edmonds' algorithm, and the data structure is constructed during the generation of the DMST.
\begin{figure}[H]
\begin{minipage}{0.48\textwidth}
	\centering
	\includegraphics[width=\linewidth]{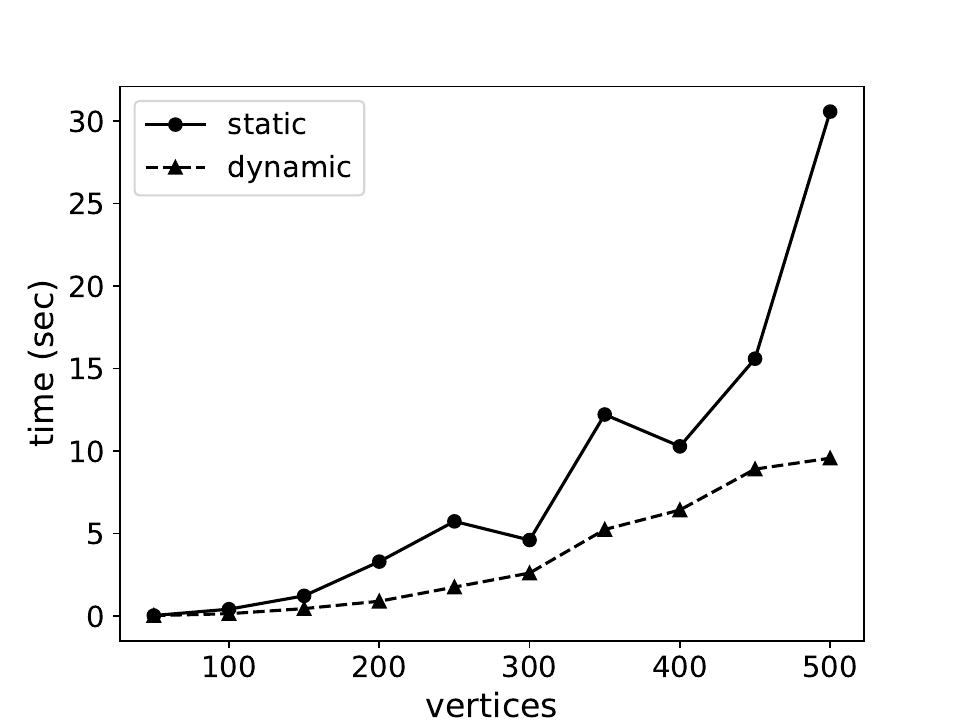}
	\caption{The performance of the edge deletion on dense graphs}\label{fig:dense_del}
\end{minipage}\hfill
\begin{minipage}{0.48\textwidth}
	\centering
	\includegraphics[width=\linewidth]{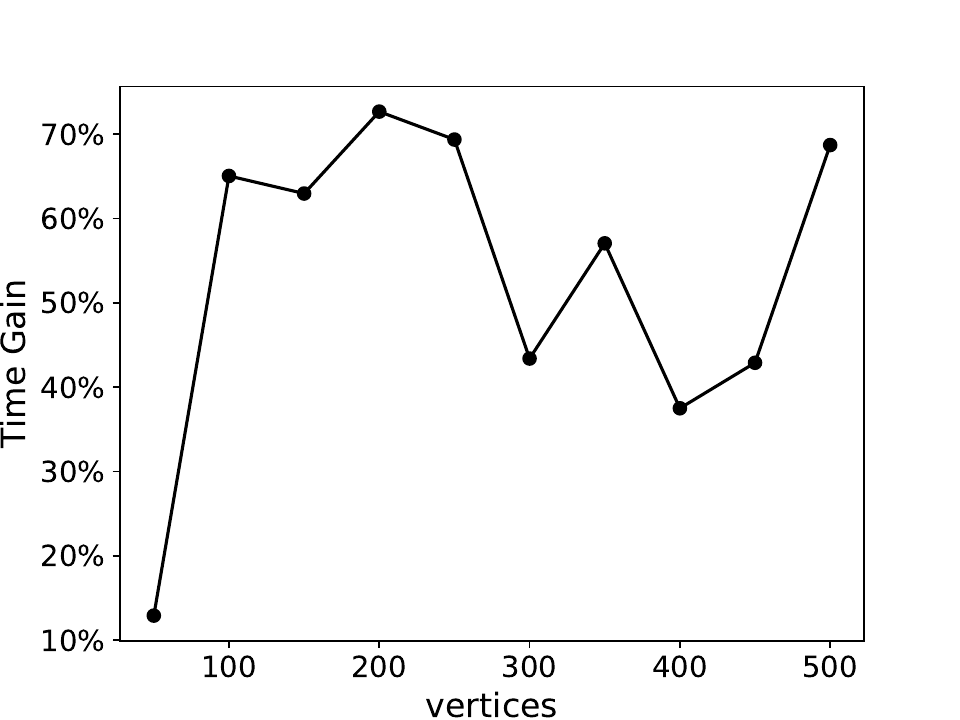}
	\caption{The time gain for the edge deletion on dense graphs}\label{fig:dense_del_gain}
\end{minipage}
\end{figure}
For the edge deletion of dense graphs, the optimization of the modified graph using our dynamic algorithm is largely faster than re-executing the static algorithm as shown in Figure \ref{fig:dense_del}. We can see in Figure \ref{fig:dense_del_gain} that the average time gain is around $50\%$.

\begin{figure}[H]
	\begin{minipage}{0.48\textwidth}
		\centering
		\includegraphics[width=\linewidth]{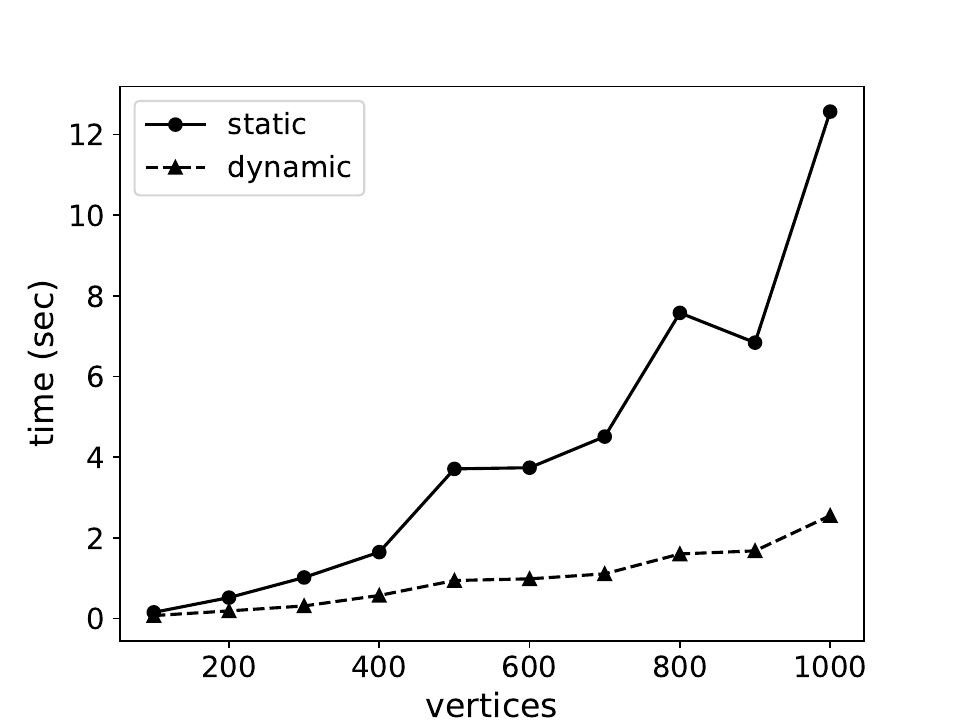}
		\caption{The performance of the edge deletion on sparse graphs}\label{fig:sparse_del}
	\end{minipage}\hfill
	\begin{minipage}{0.48\textwidth}
		\centering
		\includegraphics[width=\linewidth]{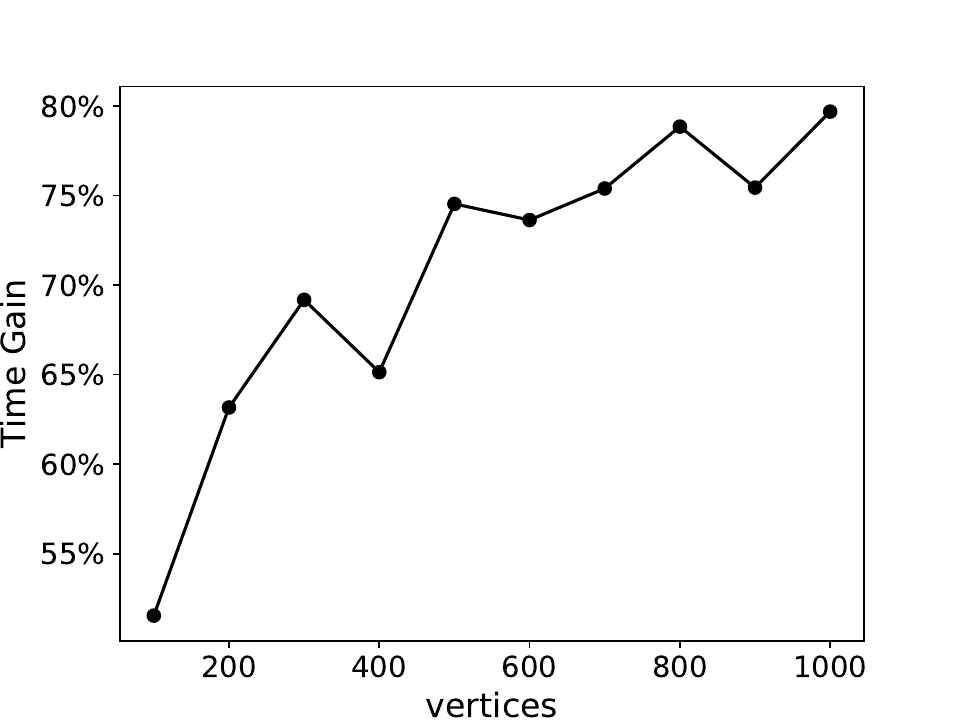}
		\caption{The time gain for the edge deletion on sparse graphs}\label{fig:sparse_del_gain}
	\end{minipage}
\end{figure}
For the edge deletion of sparse graphs, the advantage of our dynamic algorithm becomes more evident (see Figure \ref{fig:sparse_del} ). As shown in Figure \ref{fig:sparse_del_gain}, we noticed over $50\%$ average time gain in all of the instances. For the edge insertion operations, as shown in Figure \ref{fig:dense_ins_gain} and Figure \ref{fig:sparse_ins_gain}, the dynamic algorithm showed a significant time reduction compared to the static algorithm in all instances, regardless of the type of graphs. The average time gain on both dense and sparse graphs are over $75\%$ against our static matroid algorithm (see Figure \ref{fig:dense_ins} and Figure \ref{fig:sparse_ins}).

\begin{figure}[H]
	\begin{minipage}{0.48\textwidth}
		\centering
		\includegraphics[width=\linewidth]{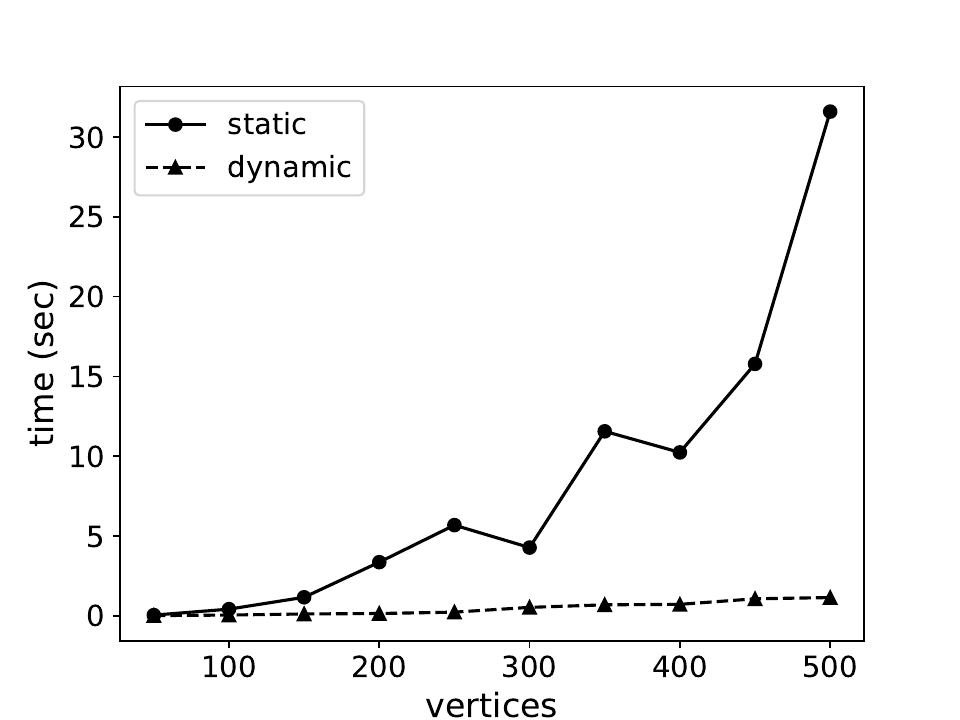}
		\caption{The performance of the edge insertion on dense graphs}\label{fig:dense_ins}
	\end{minipage}\hfill
	\begin{minipage}{0.48\textwidth}
		\centering
		\includegraphics[width=\linewidth]{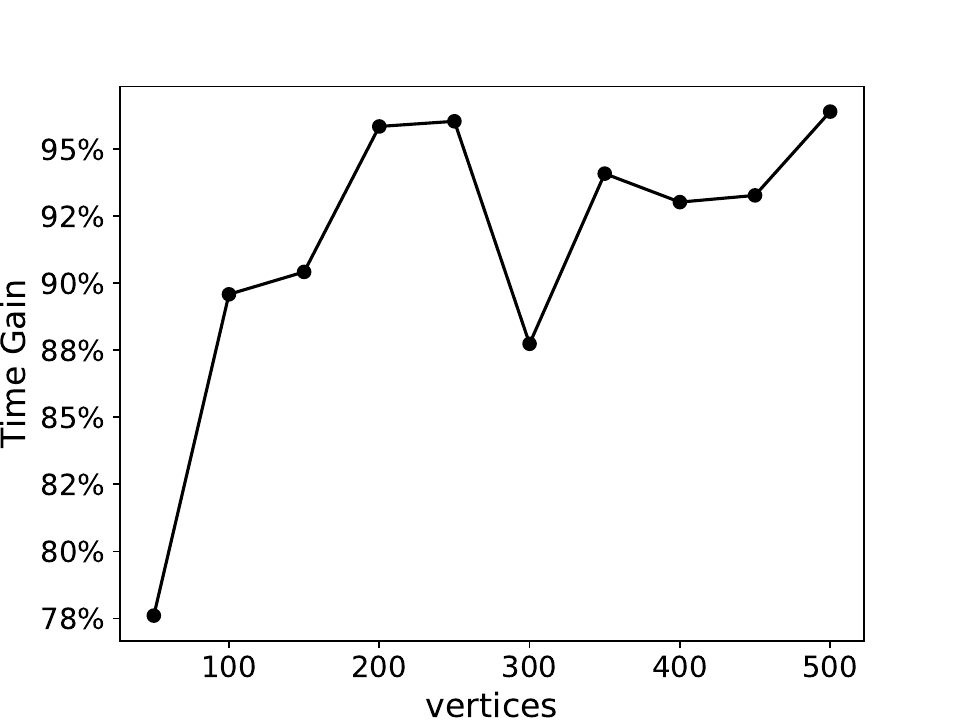}
		\caption{The time gain for the edge insertion on dense graphs}\label{fig:dense_ins_gain}
	\end{minipage}
\end{figure}

\begin{figure}[H]
	\begin{minipage}{0.48\textwidth}
		\centering
		\includegraphics[width=\linewidth]{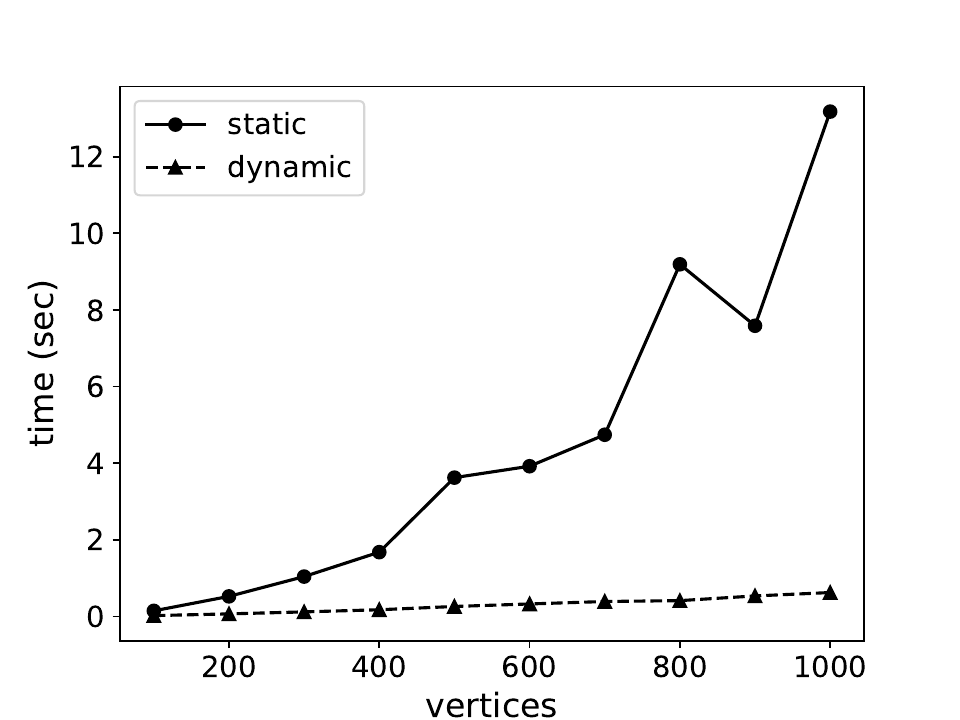}
		\caption{The performance of the edge insertion on sparse graphs}\label{fig:sparse_ins}
	\end{minipage}\hfill
	\begin{minipage}{0.48\textwidth}
		\centering
		\includegraphics[width=\linewidth]{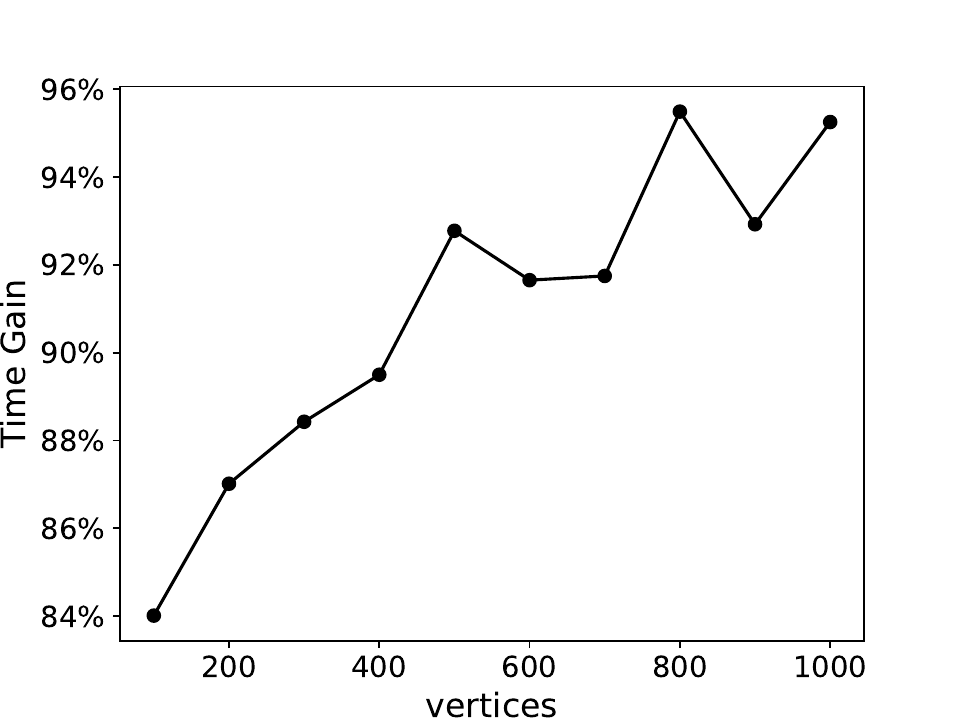}
		\caption{The time gain for the edge insertion on sparse graphs}\label{fig:sparse_ins_gain}
	\end{minipage}
\end{figure}

To sum up, the experimental evaluation of our dynamic algorithm using the matroid approach meets our expectations. By focusing on updating and maintaining the auxiliary graph, the algorithm can improve the sample tree consistently and achieve an optimal solution much faster than the static algorithm. 

\section{Conclusion}\label{S6}
In this paper, we study the problem of dynamically updating a directed spanning tree using a matroid intersection algorithm. We have designed a complete algorithm that iteratively improves sample trees — from initializing a sample tree and auxiliary graph to deriving the final optimal solution. The procedure we add to standard cycle detection — specifically, detecting simple negative cycles — is crucial, as only these cycles enable precise edge exchanges and yield improved solutions. In each iteration, we perform global cycle detection on the modified auxiliary graph. This ensures that the algorithm avoids becoming trapped in local minima. For future work, given that the auxiliary graph is bipartite, we may accelerate the detection of simple negative cycles. We also aim to further enhance performance on dense graph.

\section*{Acknowledgement}
The authors would like to thank the reviewers very much for their useful comments and suggestions, which improved this paper. The authors declare that all the data that support the results of this study are available on request from the corresponding author. All authors declare that they have no conflicts of interest.

\section*{Appendix}
\subsection*{A.1 Initiate a directed spanning tree}
To obtain an initial directed spanning tree rooted at $r$, we perform a similar greedy algorithm as Edmonds'. First, we sort the directed edge set $A$ of the graph $G=(V,A)$ in an ascending order of the weights, and let $H=\emptyset$. Then, at each step, we choose the first available edge $e$ from the sorted edge set with $h_e\neq r$ and add it to $H$, if the resulting subgraph $(V,H)$ contains no cycles, disregarding the directions, and the vertex $h_e$ has no other incoming edge. In this way, we either create a DMST when $|H|=|V|-1$, or a directed forest $(V,H)$. If we obtain a directed forest after traversing all edges, we further traverse the vertex set $V$ to find out all the vertices $v\neq r$ that has no incoming edges in $(V,H)$. For each such vertex $v$, we add a dummy edge $(r,v)$ with infinite weight to both $A$ and $H$, and denote the new graphs to be $G'=(V,A')$ and $T=(V,H')$. In this way, $T$ is our initial directed spanning tree of $G'$. Using our matroid algorithm, we can update $T$ into a DMST of the original graph $G$. We use a disjoint-set (Union-Find) data structure to maintain the direct forest. The time complexity is $O(m\log m)$, due to the sorting step being the dominant term. It can also be written as $O(m\log n)$, since $m\leq n(n-1)$. The pseudocode is the following.
\begin{algorithm}[H]
	\caption{Initial Directed Spanning Tree}
	\label{alg:greedy}
	\begin{algorithmic}[1]
		\REQUIRE Directed graph $G=(V,A)$, root vertex $r$
		\ENSURE Directed spanning tree rooted at $r$
		\STATE Let $H \gets \emptyset$
		\STATE Sort edge set $A$ by weight in an ascending order
		\STATE Initialize Union-Find structure
		\FOR{each sorted directed edge $e$}
		\IF{\textsc{Find}($h_e$) $\neq$ \textsc{Find}($t_e$), $h_e$ has no incoming edges in $H$ and $h_e \neq r$}
		\STATE \textsc{Union}($h_e$, $t_e$)
		\STATE Add $e$ to $H$
		\IF{$|H|=|V|-1$}
		\RETURN $H$
		\ENDIF
		\ENDIF
		\ENDFOR
		\FOR{each vertex $v \in V$}
		\IF{$v \neq r$ and $v$ has no incoming edges in $H$}
		\STATE Add directed edge $(r,v,+\infty)$ to $A$ and $H$
		\ENDIF
		\ENDFOR
		\RETURN $H$
	\end{algorithmic}
\end{algorithm}

\vspace{10pt} \noindent
\\
\footnotesize{\sc binhong jiang }\\
Department of Mathematics,\\
College of Information Science and Technology/College of Cyberspace Security,\\
Jinan University, Guangzhou, China. \\
\footnotesize{E-mail address:  bhjiang@stu2022.jnu.edu.cn}
\vspace{10pt} \noindent
\\
\footnotesize{\sc gehao wang }\\
Department of Mathematics,\\
College of Information Science and Technology/College of Cyberspace Security,\\
Jinan University, Guangzhou, China. \\
\footnotesize{E-mail address:  gehao\_wang@hotmail.com}


\begin{thebibliography}{99}
	\bibitem{BA} A.L. Barab\'asi, {\it Network Science}; Cambridge University Press: Cambridge, UK, 2016
	
	\bibitem{B} R.E. Bellman,  {\it On a routing problem}. Quart. Appl. Math. 16, 87-90, (1958).
	
	\bibitem{Bo} F.C. Bock, {\it An algorithm to construct a minimum directed spanning tree in a directed network}, Developments in Operations Research, pp. 29-44, Gordon and Breach, (1971).
	
	\bibitem{BCG} C. Brezovec, G. Cornuéjols  and F. Glover,  {\it Two algorithms for weighted matroid intersection}. Mathematical Programming 36, 39–53 (1986). 
	
	\bibitem{BKW} M. B\"other, O. Kißig, and C. Weyand, {\it Efficiently computing directed minimum spanning trees}, 2023 Proceedings of the Symposium on Algorithm Engineering and Experiments (ALENEX), pages 86-95, SIAM, 2023.
	
	\bibitem{BKYY} K. B\'{e}rczi, T. Kir\'{a}ly, Y. Yamaguchi and Y. Yokoi, {\it Matroid Intersection under Restricted Oracles}, SIAM Journal on Discrete Mathematics, Volume 37, Issue 2, Pages 1311-1330, 2023.
		
	\bibitem{CL}  Y. J. Chu and T. H. Liu, {\it On the shortest arborescence of a directed graph}, Scientia Sinica, 1965, 14: 1396-1400.
	
	\bibitem{CCPS} W. J. Cook, W. H. Cunningham, W. R. Pulleyblank, and A. Schrijver. {\it Combinatorial optimization}, Wiley-Interscience Series in Discrete Mathematics and Optimization, USA, 1:998, 1998.
	
   
	
	\bibitem{CG} B. Cherkassky, and A. Goldberg,  {\it Negative-cycle detection algorithms}. Math. Program. 85, 277-311 (1999).
		
		
	\bibitem{D}	E. W. Dijkstra. A note on two problems in connexion with graphs. Numerische Mathematik, pages 269-271, 1959.
		
	\bibitem{E} J. Edmonds, {\it Optimum branchings}, Journal of Research of the National Bureau of Standards, 71B (1967), p. 233.
	
	\bibitem{E1} J. Edmonds,  {\it Submodular Functions, Matroids, and Certain Polyhedra}. In: Jünger, M., Reinelt, G., Rinaldi, G. (eds) Combinatorial Optimization - Eureka, You Shrink!. Lecture Notes in Computer Science, vol 2570. Springer, Berlin, Heidelberg,  (2003).
	
	\bibitem{EFRRV} J. Espada, A. P. Francisco, T. Rocher, L. M. S. Russo, C. Vaz, {\it On Finding Optimal (Dynamic) Arborescences}. Algorithms 2023, 16, 559. https://doi.org/10.3390/a16120559.
	
	\bibitem{EGI} D. Eppstein, Z. Galil,  G.F. Italiano,  {\it 8: Dynamic graph algorithms}. In: Algorithms and Theory of Computation Handbook. CRC Presss (1999).
	
	 \bibitem{ER} P. Erdős and A. R\'enyi, {\it On random graphs}, i, Publicationes Mathematicae (Debrecen), 6 (1959).
	 
	\bibitem{FF} L.R. Ford Jr. and  D.R. Fulkerson, {\it Flows in Networks}. Princeton Univ. Press, Princeton, NJ,  (1962).
	
	\bibitem{Fr} A. Frank, , {\it A weighted matroid intersection algorithm}, Journal of Algorithms, 2(4): 328–336 (1981).
	
	\bibitem{F} S. Fujishige, {\it A primal approach to the independent assignment problem}, Journal of the Operations Research Society of Japan 20 (1977) 1-15.
	
	
	\bibitem{FO} O. Fischer and R. Oshman. {\it  A distributed algorithm for directed minimum weight spanning tree}, Distributed Computing, 36(1):57-87, 2023.

	
	\bibitem{GGST} H. N. Gabow, Z. Galil, T. H. Spencer, and R. E. Tarjan, {\it Efficient algorithms for finding minimum spanning trees in undirected and directed graphs}, Combinatorica, 6 (1986).
	
	\bibitem{HHS} K. Hanauer, M. Henzinger, and C. Schulz.  {\it Recent Advances in Fully Dynamic Graph Algorithms – A Quick Reference Guide}. ACM J. Exp. Algorithmics 27, Article 1.11 (December 2022), 45 pages, 2022.
	
	\bibitem{HSS} A. Hagberg P. J. Swart  and D. A. Schult,
	Exploring network structure, dynamics, and function using NetworkX, Report,
	Los Alamos National Laboratory (LANL), Los Alamos, NM (United States), 2008.
	
	\bibitem{HT} D. Harel, R. E. Tarjan, {\it Fast algorithms for finding nearest common ancestors}. SIAM J. Comput. 13(2), 338–355 (1984)
	
	
	\bibitem{K} R. Karp, {\it A simple derivation of Edmonds' algorithm for optimum branchings}, Networks 1 (1971): 265-272.
	
		
	\bibitem{KN} N. Kamiyama, {\it Arborescence problems in directed graphs: Theorems and algorithms}, Interdiscip. Inform. Sci. 20 (2014) 51-70.
	
	\bibitem{KJ} B. Korte, and J. Vygen. {\it Combinatorial optimization: theory and algorithms}, 6th edition, Springer, New York.
	
	\bibitem{La} E. L. Lawler, {\it Matroid intersection algorithms}, Mathematical Programming, 9: 31-56 (1975).
	
	\bibitem{L} L. Lov\'asz, {\it Computing ears and branchings in parallel}, in 26th Annual Symposium on Foundations of Computer Science (sfcs 1985), 1985, pp. 464-467
	
	\bibitem{MARRB} A. Madkour,  W. G. Aref,   F. U. Rehman,  M. A. Rahman, and S. Basalamah, (2017). {\it A survey of shortest-path algorithms}. arXiv:1705.02044.
	
	
	\bibitem{M} E.F. Moore, {\it The Shortest Path Through a Maze}. In: Proc. of the Int. Symp. on the Theory of Switching, pp. 285-292, (1959), Harvard University Press.
	
	
	\bibitem{O} J. G. Oxley , {\it Matroid Theory}, Oxford University Press, 2nd edition (2011).
	
	\bibitem{PMW} J.H. Pan, T. Mitra and W-F. Wong, {\it Configuration bitstream compression for dynamically reconfigurable FPGAs}, IEEE/ACM International Conference on Computer Aided Design, 2004. ICCAD-2004., San Jose, CA, USA, 2004, pp. 766-773, doi: 10.1109/ICCAD.2004.1382679.
	
	\bibitem{PTZ} G. G. Pollatos, O. A. Telelis,  V. Zissimopoulos, {\it Updating Directed Minimum Cost Spanning Trees}. In: C. \`Alvarez, M. Serna, (eds) Experimental Algorithms. WEA 2006. Lecture Notes in Computer Science, vol 4007. Springer, Berlin, Heidelberg.
	
	\bibitem{S} A. Schrijver, {\it Combinatorial Optimization: Polyhedra and Efficiency}, Springer Berlin, Heidelberg, 2003. 
	
	\bibitem{SJ}K. S\"orensen, and G. Janssens. {\it An algorithm to generate all spanning trees of a graph in order of increasing cost}, Pesquisa Operacional, 25 (2005): 219-229.
	
	
	\bibitem{T} R. E. Tarjan, {\it Finding optimum branchings}, Networks, 7(1):25-35, 1977.
	
	 \bibitem{T2} R. E. Tarjan,  {\it Shortest Paths}. Technical report, AT\&T Bell Laboratories, Murray Hill, NJ, (1981).

    \bibitem{XW} L. Xu, D. Wen, L. Qin, R. Li, Y. Zhang, Y. Lu, and X. Lin. {\it Minimum Spanning Tree Maintenance in Dynamic Graphs}. Proc. ACM Manag. Data 3, 1, Article 54 (February 2025), 24 pages. https://doi.org/10.1145/3709704

  


\end{thebibliography}
\end{document}